\documentclass[11pt]{article}
\usepackage[utf8]{inputenc}
\usepackage{hyperref}
\usepackage{amsmath}
\usepackage{amsthm}
\usepackage{amssymb}
\usepackage{authblk}
\usepackage{cite}
\usepackage{lipsum}
\usepackage{geometry}
\usepackage{enumerate}
\usepackage{array}
\usepackage{multicol}
\usepackage{multirow}
\usepackage{xcolor}
\usepackage{slashbox}
\newtheorem{theorem}{Theorem}[section]
\numberwithin{theorem}{section}
\newtheorem{lemma}[theorem]{Lemma}
\newtheorem{defi}[theorem]{Definition}
\newtheorem{corollary}[theorem]{Corollary}

\newtheorem{remark}[theorem]{Remark}

\DeclareMathOperator{\tr}{Tr}

\newcommand{\F}{\mathbb F}
\newcommand{\Fbn}{\mathbb{F}_{2^n}}

\newcommand{\Fbnmul}{\mathbb{F}_{2^n}^*}

\newcommand{\Fbnn}{\mathbb{F}_{2^n} \longrightarrow \mathbb{F}_{2^n}}

\begin{document}

    \title{Investigations of $c$-Differential Uniformity of Permutations with Carlitz Rank 3}
    \author{Jaeseong Jeong$^1$, Namhun Koo$^2$, Soonhak Kwon$^1$\\
        \small{\texttt{ Email: wotjd012321@naver.com, nhkoo@ewha.ac.kr, shkwon@skku.edu}}\\
\small{$^1$Applied Algebra and Optimization Research Center, Sungkyunkwan University, Suwon, Korea}\\
\small{$^2$Institute of Mathematical Sciences, Ewha Womans University, Seoul, Korea}\\
    }

\maketitle
\begin{abstract}
The $c$-differential uniformity is recently proposed to reflect resistance against some variants of differential attack. Finding functions with low $c$-differential uniformity is attracting attention from many researchers. For even characteristic, it is known that permutations of low Carlitz rank have good cryptographic parameters, for example, low differential uniformity, high nonlinearity, etc. In this paper we show that permutations with low Carlitz rank have low $c$-differential uniformity. We also investigate $c$-differential uniformity of permutations with Carlitz rank 3 in detail.

\bigskip
\noindent \textbf{Keywords.} $c$-Differential Uniformity, Carlitz rank, Permutation,

\bigskip
\noindent \textbf{Mathematics Subject Classification(2020)} 94A60, 06E30
\end{abstract}

\section{Introduction}

Functions with good cryptographic parameters have many applications in cryptographic purpose.
Very recently, Ellinsen et. al.\cite{EFR+20} proposed a new cryptographic parameter, $c$-differential uniformity, which is useful to measure the resistance against some variants of differential attack\cite{BCJW02}. Functions with low $c$-differential uniformity are resistant against differential attacks of this type. Finding functions with low $c$-differential uniformity has been a good research topic to many researchers, and many classes of functions with low $c$-differential uniformity were proposed.\cite{BC21,BCRP21,HPR+21,MRS+21,Sta21b,Sta21d,TZJT21,WZ21,WLZ21,Yan21,ZH21,ZY22}
It is known \cite{MRS+21} that some functions of low differential uniformity have high $c$-differential uniformity. Hence it is also important to investigate $c$-differential uniformity of known functions with low differential uniformity.

For even characteristic, some cryptographic properties of permutations of low Carlitz rank have been investigated in several researches, and we can see that they have good cryptographic parameters. For example, the multiplicative inverse function of Carlitz rank 1 has low differential uniformity\cite{Nyb94}, high nonlinearity\cite{LW90}, low boomerang uniformity\cite{BC18}, low differential-linear uniformity\cite{CKL+21}, low $c$-differential uniformity\cite{EFR+20}, and low $c$-boomerang uniformity\cite{Sta21a}. Furthermore, it is used as the S-box of the AES(Advanced Encryption Standard) cryptosystem. Cryptographic parameters of permutations of Carlitz rank 2 also have been widely investigated. It is known that they have low differential uniformity\cite{LWY13e}, high nonlinearity\cite{LWY13e}, low boomerang uniformity\cite{LQSL19}, low differential-linear uniformity\cite{JKK22a}, low $c$-differential uniformity\cite{Sta21b}, and low $c$-boomerang uniformity\cite{Sta21b}. Several classes of differentially $4$-uniform involutions with low Carlitz rank were proposed in \cite{JKK22b}.

Permutations with Carlitz rank 3 also have good cryptographic parameters, for example, high nonliearity\cite{LWY13e}, and low differential-uniformity\cite{JKK22a}. Differentially $4$-uniform permutations with Carlitz rank 3 are characterized in \cite{LWY13e, JKK20}. The boomerang uniformities of permutations with Carlitz rank 3 were investigated in \cite{JKK20}. Permutations with low Carlitz rank are known to have low differential uniformity, but it does not imply low $c$-differential uniformity as mentioned above. 

In this paper, we show that for binary finite fields the $c$-differential uniformity of permutations with Carlitz rank $m$ is upper bounded by $m+2$ with $c\ne 1$. Furthermore, we investigate $c$-differential uniformity of permutations of Carlitz rank 3 in detail. 

The rest of this paper is organized as follows. In section 2, we introduce some basic preliminaries and previous results which are necessary in subsequent sections.
In section 3, we propose an upper bound of $c$-differential uniformity depending on Carlitz rank of given permutation. We investigate $c$-differential uniformity of permutations with Carlitz rank $3$ in section 4. Finally we give the concluding remark in Section 5.

\section{Preliminaries}\label{sec_pre}

We only consider the even characteristic case. Throughout this paper, we let :
\begin{itemize} 
\item $\Fbn$ be the finite field of $2^n$ elements and $\Fbnmul$ be the multiplicative group of $\Fbn$
\item $\tr : \Fbn \longrightarrow \F_2$ be the field trace from $\Fbn$ onto $\F_2$ given by $\tr(x)=x+x^2+x^{2^2}+\cdots+ x^{2^{n-1}}$
\item $Inv$ be the multiplicative inverse function on $\Fbn$, and $x^{-1}=Inv(x)$ for all $x\in\Fbnmul$ and $0^{-1}=0$
\end{itemize}

Next we introduce $c$-differential uniformity which is the subject of this paper.

\begin{defi}(\cite{EFR+20}) Let $F : \Fbnn$ be a function and $c\in\Fbn$.\\
(i) We denote the $c$-differential of $F$ by $_cD_a F(x)=F(x+a)-cF(x)$.\\
(ii) Let $a,b\in\Fbn$. We denote $_c\Delta_F(a,b)$ by the number of solutions in $\Fbn$ of $_cD_a F(x)=b$.\\
(iii) The $c$-differential uniformity of $F$ is defined by $_c \Delta_F = \max\{ {}_c\Delta_F(a,b) : a,b\in \Fbn\text{ and }a\ne 0 \text{ if } c=1 \}$.\\
(iv) $F$ is called a perfect $c$-nonlinear(PcN) function if $_c \Delta_F =1$.\\
(v) $F$ is called a almost perfect $c$-nonlinear(APcN) function if $_c \Delta_F =2$.
\end{defi}

\medskip
 It is known \cite{LN15} that, for any permutation $F : \Fbnn$, there is $m \geq 0$ and $a_i \in \Fbn \,\, (0\leq i \leq
m+1)$ such that
\begin{equation}\label{car}
F(x) = (\cdots((a_0x + a_1)^{2^n-2}+a_2)^{2^n-2} \cdots +a_m)^{2^n-2} + a_{m+1},
\end{equation}
where $a_0, a_2, \cdots , a_{m} \neq 0$.
 The above expression means that any permutation on $\Fbn$ is generated by the inverse function $x^{2^n-2}$ and linear functions $ax+b\,\, (a\neq 0)$. The \emph{Carlitz rank} of $F$ is the least nonnegative integer $m$ satisfying the above expression\cite{ACMT09}.

It is easy to see that if $F$ has the Carlitz rank $1$ then $F$ is affine equivalent to $Inv$. The $c$-differential uniformity of $Inv$ was investigated in \cite{EFR+20}.

\begin{theorem}\label{cdu_inv_even}(\cite{EFR+20}) Let  $c\ne 0,1$. Then $Inv$ is APcN if and only if $\tr(c)=\tr(1/c)=1$. Otherwise, $ _c\Delta_{Inv}=3$.
\end{theorem}

The $c$-differential uniformity of $Inv \circ (0,1)$ was investigated in \cite{Sta21b}, where $(0,1)$ is the transposition that $0$ and $1$ are swapped. As observed in \cite{JKK20} that this case is related with permutations of Carlitz rank $2$.

\begin{theorem}\label{cdu_car2}(\cite{Sta21b}) Let $c\ne 0,1$ and $F=Inv \circ (0,1)$ on $\Fbn$.\\
(i) If $n=2$, then $_c\Delta_F =1$.\\
(ii) If $n=3$, then $_c\Delta_F \le 3$.\\
(iii) If $n\ge 4$, then $_c\Delta_F\le 4$.
\end{theorem}

Next we introduce a well-known lemma about the number of solutions of quadratic equations.

\begin{lemma}\label{quad_lemma_even}
Let $p=2$, $a_2\in\Fbnmul$ and $a_1,a_0\in\Fbn$. Then,
\begin{equation*}
\#\{x\in\Fbn : a_2x^2+a_1x+a_0=0\}=
\begin{cases}
2 &\text{ if } a_1\ne0\text{ and }\tr\left(\frac{a_0a_2}{a_1^2}\right)=0,\\
1 &\text{ if } a_1=0,\\
0 &\text{ if } a_1\ne0\text{ and }\tr\left(\frac{a_0a_2}{a_1^2}\right)=1.\\
\end{cases}
\end{equation*}
\end{lemma}

The following lemma is very simple and useful to our results.

\begin{lemma}\label{cdu_symm_lemma} Let $c\ne 0$. Then $_c\Delta_F(a,b)=\ _{c^{-1}}\Delta_F(a,bc^{-1})$ for all $a\in\Fbnmul$ and $b\in\Fbn$. If $F$ is a permutation then $_c\Delta_F={}_{c^{-1}}\Delta_F$.
\end{lemma}
\begin{proof}
Let $a\in\Fbnmul$ and $b\in\Fbn$. Then, it is easy to see that $x$ is a solution of $_c D_aF(x)=b$ if and only if $x+a$ is a solution of $_{c^{-1}}D_a(x)=bc^{-1}$. Thus we have $_c\Delta_F(a,b)={} _{c^{-1}}\Delta_F(a,bc^{-1})$. If $F$ is a permutation then $_c\Delta_F(0,b)=1$ for all $b\in\Fbn$ and hence we get $_c\Delta_F={} _{c^{-1}}\Delta_F$.
\end{proof}

\section{Upper bound on $c$-differential uniformity of permutations with low Carlitz rank}

There are many differentially $4$-uniform permutations obtained from modifying some points in the multiplicative inverse function(see \cite{JKK22a} and its references) defined on $\Fbn$. In this section we give an upper bound of $c$-differential uniformity of such differentially $4$-uniform permutations when $c\ne 1$.

\begin{theorem}\label{cdu_upper} Let $F$ and $G$ be permutations on $\Fbn$ such that $F(x)=G(x)$ for all $x\in\Fbn\setminus P$ for some nonempty $P\subseteq \Fbn$. If $c\ne 1$ then $_c\Delta_F\le {}_c\Delta_G+ \# P$.
\end{theorem}
\begin{proof} Let $P_a=P\cup\{x+a : x\in P\}$. It is enough to show that $_c\Delta_F(a,b)\le \frac{\#P_a}{2}+{}_c\Delta_F(a,b)$ for all $a\in\Fbnmul$ and $b\in\Fbn$.
Then for $x\in\Fbn \setminus P_a$, all solutions of $_cD_aF(x)=b$ are also solutions of  $_cD_aG(x)=b$ since $_cD_aF(x)={}_cD_aG(x)$. Thus we have $_cD_aF(x)=b$ has at most $_c \Delta_G$ solutions in $\Fbn \setminus P_a$.\\
Let $x\in P_a$. If $_cD_aF(x)=F(x+a)+cF(x)=F(x)+cF(x+a)={}_cD_aF(x+a)$ then we have $0=F(x+a)+cF(x)+F(x)+cF(x+a)=(c+1)(F(x+a)+F(x))$. Since $F$ is a permutation, we have $c=1$, which contradicts to the assumption.
This means that 
\begin{equation}\label{asa_eqn}
_cD_aF(x)\ne {}_cD_aF(x+a)\text{ for all }x\in P_a.
\end{equation}
 If there are $\frac{\#P_a}{2}+1$ solutions in $P_a$ of $_c D_a F(x)=b$ for some $a, b\in\Fbn$ then there is $x\in P_a$ such that both $x$ and $x+a$ are solutions of $_c D_a F(x)=b$, a contradiction to $_cD_aF(x)\ne {}_cD_aF(x+a)$ for all $x\in P_a$. Hence there are at most $\frac{\#P_a}{2}$ solutions in $P_a$ of $_c D_a F(x)=b$.\\
Therefore, $_c\Delta_F(a,b)\le {}_c \Delta_G(a,b)+\frac{\#P_a}{2}$ for all $a,b\in\Fbn$, which completes the proof.
\end{proof}

\noindent \textbf{Example.} Let $F=Inv\circ (1,\gamma)$ where $\gamma\in\Fbn\setminus \F_2$. The differential uniformity of this function is characterized in \cite{LWY13e}. We have $G=Inv$ and $P=\{1,\gamma\}$ and hence we get $_c\Delta_F \le 5$ by Theorem \ref{cdu_inv_even} and Theorem \ref{cdu_upper_inv} when $c\ne 1$.\\

By the above theorem, we can see that if $G$ has a low $c$-differential uniformity and $\#P$ is small then $F$ has also a low $c$-differential uniformity. Next theorem shows that this upper bound can be slightly reduced when $G=Inv$ and $0\in P$.

\begin{theorem}\label{cdu_upper_inv} Let $F$ be a permutation on $\Fbn$ such that $F(x)=Inv(x)$ for all $x\in\Fbn\setminus P$ for some nonempty $P\subseteq \Fbn$ with $0\in P$. If $c\ne 1$ then $_c\Delta_F\le \# P+2$.
\end{theorem}
\begin{proof} Let $P_a=P\cup\{x+a : x\in P\}$. It is enough to show that
\begin{equation}\label{car3cdu_eqn}
_c\Delta_F(a,b)\le \frac{\#P_a}{2}+2
\end{equation}
for all $a\in\Fbnmul$ and $b\in\Fbn$. By the same argument in Theorem \ref{cdu_upper}, there are at most $\frac{\#P_a}{2}$ solutions in $P_a$ of $_c D_a F(x)=b$. For $x\in\Fbn \setminus P_a$, we have $_cD_aF(x)={}_cD_a Inv(x)$. Since $0,a\in P_a$, $_cD_aF(x)=b$ implies $bx^2+(ab+c+1)x+ca=0$ which has at most two solutions. Therefore, we have $_c\Delta_F(a,b)\le\  \frac{\#P_a}{2}+2$ for all $a,b\in\Fbn$, which completes the proof.
\end{proof}

It is known \cite{HPR+21} that $F$ and $F\circ A$ has the same $c$-differential spectrum, where $A$ is an affine permutation. However, it is not trivial to show $F$ and $A\circ F$ have the same $c$-differential uniformity, rather, it seems that it does not generally hold. For example, if $\F_{2^4}=\F_2[X]/\langle X^4+X+1\rangle$ and $g^4+g+1=0$ and $A(x)=x^4+gx$, then $A$ is an affine permutation. Then we have $_g\Delta_{A\circ Inv}=4$ and $_g\Delta_{Inv}=3$ and hence $_g\Delta_{A\circ Inv}\ne {}_g\Delta_{Inv}$. However we show that the $c$-differential uniformity is also invariant in some specific affine equivalence. For convenience we call $A(x)=ux+v$ an \emph{affine permutation of degree one} if $u\in \Fbnmul$ and $v\in\Fbn$ (here \emph{degree} does not mean \emph{algebraic degree}). We also call two permutations $F$ and $F'$ are \emph{affine equivalent of degree one} if there are affine permutations $A_1$ and $A_2$ of degree one such that $F=A_1 \circ F' \circ A_2$. It is clear that affine equivalence of degree one is an equivalence relation. Next we show that two permutations which are affine equivalent of degree one has the same $c$-differential uniformity. 

\begin{lemma}\label{cdu_affine_lemma} Let $F$ and $F'$ be permutations on $\Fbn$ which are affine equivalent of degree one. Then $_c\Delta_F =\ _c\Delta_{F'}$ for all $c\in\Fbn$.
\end{lemma}
\begin{proof}
We denote $F=A_1 \circ F' \circ A_2$ and $F''=F' \circ A_2$, where $A_1$ and $A_2$ are affine permutations of degree one. Then we already see that $_c\Delta_{F''}={}_c\Delta_{F'}$ in \cite{HPR+21}. Hence it is sufficient to show that $_c\Delta_{F''}={} _c\Delta_{F}$. We denote $A_1(x)=u_1 x+v_1$ where $u_1\in\Fbnmul$ and $v_1\in\Fbn$. Let $a,b\in \Fbn$. Then $b={} _c D_a F(x)$  implies that 
\begin{align*}
b&=F(x+a)+cF(x)=(A_1 \circ F'')(x+a)+c(A_1\circ F'')(x)=A_1(F''(x+a))+cA_1(F''(x))\\
&=u_1 F''(x+a)+v_1 +c(u_1 F''(x)+v)=u_1(F''(x+a)+cF''(x))+(c+1)v_1
\end{align*}
Hence we have $_c D_a F''(x)=u_1^{-1}b+u_1^{-1}v_1(c+1)$. It is clear that a map from $b$ to $u_1^{-1}b+u_1^{-1}v_1(c+1)$ is bijective on $\Fbn$ for all $c\in\Fbn$ and therefore we have $_c\Delta_{F''}={}_c\Delta_{F}$.
\end{proof}

Now let $F$ be defined as \eqref{car}. For recurrence relations $\alpha_i=a_i\alpha_{i-1}+\alpha_{i-2}$ and $\beta_i=a_i\beta_{i-1}+\beta_{i-2}$ with $\alpha_0=0, \alpha_1=a_0, \beta_0=1, \beta_1= a_1$ where $2\le i\le m+1$, we denote
\begin{equation*}
R_m(x)=\frac{\alpha_{m+1}x+\beta_{m+1}}{\alpha_m x+\beta_m}
\end{equation*}
and $O_m=\left\{x_i : x_i =\frac{\beta_i}{\alpha_i}, 1\le i \le m\right\}$. Then it is known \cite{ACMT09} that $F(x)=R_m(x)$ for all $x\in\Fbn\setminus O_m$. Next we recall Lemma 3.3 of \cite{JKK20} with some additional properties we need.

\begin{lemma}\label{car_affine_lemma} Let $F$ be defined as \eqref{car} and $\alpha_i, \beta_i$ defined as above.\\
(i) If $\alpha_m\ne 0$ then  there is a permutation $G$ such that $G$ is affine equivalent of degree one to $F$ and $G(x)=Inv(x)$ for all $x\in\Fbn\setminus P$ where $P\subseteq \Fbn$ with $\#P\le m$ and $0\in P$.\\
(ii) If $\alpha_m=0$ then  there is a permutation $G$ such that $G$ is affine equivalent of degree one to $F$ and $G(x)=x$ for all $x\in\Fbn\setminus P$ where $P\subseteq \Fbn$ with $\#P\le m$.
\end{lemma}
\begin{proof}
For (i) we denote $A_1(x) = \frac{a_0 x+\beta_m}{\alpha_m}$ and $A_2(x)=\alpha_m x+\alpha_{m+1}$ then $A_1$ and $A_2$ are affine permutations of degree one. Let $P=\{y_i : y_i=A_1^{-1}(x_i), x_i \in O_m\}$ then $\#P\le m$. Since $\alpha_i \beta_{i+1}+\alpha_{i+1}\beta_i = \alpha_i(a_i \beta_i+\beta_{i-1})+\beta_i(a_i\alpha_i+\alpha_{i-1})=\alpha_{i-1}\beta_i+\alpha_i\beta_{i-1}$, we have $\alpha_m\beta_{m+1}+\alpha_{m+1}\beta_{m}=\alpha_0\beta_{1}+\alpha_{1}\beta_{0}=a_0$ recursively. Using $a_0=\alpha_m\beta_{m+1}+\alpha_{m+1}\beta_{m}$, for all $x\in \Fbn \setminus P$
\begin{align*}
(F\circ A_1)(x)&=R_m\left( A_1(x)\right)=\frac{\alpha_{m+1}\cdot \frac{a_0 x+\beta_m}{\alpha_m}+\beta_{m+1}}{\alpha_m \cdot \frac{a_0 x+\beta_m}{\alpha_m}+\beta_m}=\frac{\alpha_{m+1}(a_0x+\beta_m)+\alpha_m\beta_{m+1}}{a_0\alpha_m x}\\
&=\frac{a_0\alpha_{m+1}x+(\alpha_{m+1}\beta_m+\alpha_m\beta_{m+1})}{a_0\alpha_m x}=\frac{a_0\alpha_{m+1}x+a_0}{a_0\alpha_m x}=\frac{1}{\alpha_m}\left(\alpha_{m+1}+\frac{1}{x}\right)\\
(A_2\circ F\circ A_1)(x)&=A_2\left((F\circ A_1)(x)\right)=\alpha_m\cdot \frac{1}{\alpha_m}\left(\alpha_{m+1}+\frac{1}{x}\right)+\alpha_{m+1}=\frac{1}{x}=Inv(x).
\end{align*}
Since $A_1(0)=\frac{\beta_m}{\alpha_m}$, we can see that $0=A_1^{-1}\left(\frac{\beta_m}{\alpha_m}\right)\in P$. The proof of (ii) is clear by substituting $\alpha_m=0$ in $R_m(x)$.
\end{proof}

Now we are ready to get an upper bound on $c$-differential uniformity of permutations with Carlitz rank $m$.

\begin{theorem} Let $F$ be a permutation on $\Fbn$ with Carlitz rank $m$, and $c\ne 1$. Then $_c \Delta_F\le m+2$.
\end{theorem}

\begin{proof}
The case $\alpha_m\ne 0$ is directly from Theorem \ref{cdu_upper_inv}, Lemma \ref{cdu_affine_lemma} and Lemma \ref{car_affine_lemma} (i). If $\alpha_m=0$ then there is a permutation $G$ such that $G(x)=x$ for all $x\in\Fbn\setminus P$ with $\#P\le m$, by Lemma \ref{car_affine_lemma} (ii). Thus $_c D_a G(x)=b$ has at most $\frac{\#P_a}{2}$ solutions in $P_a=P\cup \{x+a : x\in P\}$ similarly with Theorem \ref{cdu_upper}. For $x\in \Fbn\setminus P_a$, $_c D_a G(x)=b$ implies $(c+1)x=a+b$ which has exactly one solution if $c\ne 1$. Thus we have $_c\Delta_G(a,b)\le \frac{\#P_a}{2}+1\le m+1$. Therefore we have $_c\Delta_G \le m+1\le m+2$, which completes the proof, by Lemma \ref{cdu_affine_lemma}.
\end{proof}

It is easy to see that permutations of Carlitz rank $1$ are affine equivalent of degree one to $Inv$. We observed in \cite{JKK20} that permutations of Carlitz rank $2$ are affine equivalent of degree one to $Inv\circ (0,1)$. Using Lemma \ref{cdu_affine_lemma}, we can see that the upper bound in the above theorem is tight when $m=1$ and $m=2$ by Theorem \ref{cdu_inv_even} and Theorem \ref{cdu_car2}, respectively. In next section, we show that it is also tight when $m=3$.

\section{$c$-differential uniformity of permutations with Carlitz rank 3}

Let $F$ be \eqref{car} with $m=3$ as follows :
\begin{equation}\label{carform3}
F(x)=(((a_0x+a_1)^{2^n-2}+a_2)^{2^n-2}+a_3)^{2^n-2}+a_4
\end{equation}
 If $A_1(x)=\frac{x+a_1a_2}{a_0a_2}$ and $A_2(x)=\frac{x+a_4}{a_2}$ then we can easily check 
$$F_3(x)=(A_2 \circ F\circ A_1)(x)=((x^{2^n-2}+1)^{2^n-2} +a_2a_3)^{2^n-2}.$$
Now we denote $\gamma=a_2a_3$. If $\gamma=1$, then $F_3$ can be expressed by 
\begin{equation*}
F_3(x)=
\begin{cases}
x+1&\text{if }x\not\in\{0,1\}\\
x&\text{if }x\in\{0,1\}
\end{cases}
\end{equation*}
which is not an interesting case. When $\gamma\ne 0,1$, as observed in \cite{JKK20}, if we denote $A_3(x)=\frac{x+\gamma}{\gamma+1}$ and $A_4(x)=(\gamma+1)x+1$ then we have $A_4 \circ F_3 \circ A_3 = Inv \circ (0,1,\gamma)$ where $(0,1,\gamma)$ is a cycle of length $3$. Thus we can see that $F$ in \eqref{carform3} is affine equivalent of degree one to $Inv\circ(0,1,\gamma)$ with $\gamma=a_2a_3$. From now on, we investigate $c$-differential uniformity of $Inv\circ(0,1,\gamma)$ where $\gamma\in\Fbn\setminus\F_2$. 

\begin{remark} For small $n$, we obtain the following results from the exhaustive search.\\
(i) If $n=2$ then $F$ is PcN for all $c\in\F_4\setminus \F_2$\\
(ii) If $n=3$ with $\F_{2^3}=\F_2[x]/\langle x^3+x+1\rangle$ then we have	
\begin{equation*}
_c\Delta_F=
\begin{cases}
2 &\text{ if }\gamma^3+\gamma+1=0,\\
3 &\text{ otherwise.}
\end{cases}
\end{equation*}
\end{remark}

We investigate the $c$-differential uniformity of $Inv\circ(0,1,\gamma)$ when $n\ge 4$. If $c=0$ then $_0 \Delta_F=1$ since $F$ is a permutation. The case $c=1$ was already investigated in \cite{LWY13e,JKK20}. Hence we assume that 
\begin{equation}\label{car3}
F=Inv \circ (0,1,\gamma),\  c,\gamma\not\in\{0,1\},\ n\ge4
\end{equation}
throughout this section unless otherwise noted. By Theorem \ref{cdu_upper_inv}, $F$ has $c$-differential uniformity at most $5$. We investigate $c$-differential uniformity of $F$ in detail. We denote a set
$$P_a = P\cup \{x+a : x\in P\}= \{0,1,\gamma, a,a+1,a+\gamma \}.$$
We first characterize the condition that $_c D_aF(x)=b$ has two solutions in $\Fbn\setminus P_a$.

\begin{lemma}\label{panot2_lemma} Let $a\in\Fbnmul$ and $b\in\Fbn$. Then $_c D_aF(x)=b$ has at most two solutions in $\Fbn\setminus P_a$. Furthermore, $_c D_aF(x)=b$ has two solutions in $\Fbn\setminus P_a$ if and only if $b\ne 0$, $ab+c+1\ne 0$, $\tr\left(\frac{abc}{(ab+c+1)^2}\right)=0$, $(b+c)a+b+c+1\ne0$, $(b+1)a+b+c+1\ne0$, $(b\gamma+c)a+\gamma(b\gamma+c+1)\ne0$ and $(b\gamma+1)a+\gamma(b\gamma+c+1)\ne0$.
\end{lemma}
\begin{proof}
If $x\in\Fbn\setminus P_a$ then $b={}_cD_aF(x)={}_cD_a Inv(x)$ implies that $bx^2+(ab+c+1)x+ac=0$ which has at most $2$ solutions. Hence $_c D_aF(x)=b$ has at most two solutions in $\Fbn\setminus P_a$.\\
$_cD_aF(x)=b$ has $2$ solutions in $\Fbn\setminus P_a$ implies that $bx^2+(ab+c+1)x+ac=0$ has $2$ solutions which is equivalent to
\begin{equation}\label{cdu5_trace}
\tr\left(\frac{abc}{(ab+c+1)^2}\right)=0,\ b\ne 0\text{ and }ab+c+1\ne0.
\end{equation} 
by Lemma \ref{quad_lemma_even}. If there is a solution $x_0\in P_a$ of $bx^2+(ab+c+1)x+ac=0$ then $x_0$ cannot be a solution of $_cD_aF(x)=b$. Hence we need to check that there is no solutions of $bx^2+(ab+c+1)x+ac=0$ in $P_a$.
\begin{itemize}
\item If $x=0$ or $x=a$ then we have $a=0$ which is a contradiction to $a\in\Fbnmul$.
\item If $x=1$ then we have $(b+c)a+b+c+1=0$ or equivalently $a=\frac{b+c+1}{b+c}$.
\item If $x=a+1$ then we have $(b+1)a+b+c+1=0$ or equivalently $a=\frac{b+c+1}{b+1}$.
\item If $x=\gamma$ then we have $(b\gamma+c)a+\gamma(b\gamma+c+1)=0$ or equivalently $a=\frac{\gamma(b\gamma+c+1)}{b\gamma+c}$.
\item If $x=a+\gamma$ then we have $(b\gamma+1)a+\gamma(b\gamma+c+1)=0$ or equivalently $a=\frac{\gamma(b\gamma+c+1)}{b\gamma+1}$.
\end{itemize}
Hence we can see that there are no solutions in $P_a$ of $bx^2+(ab+c+1)x+ac=0$ if and only if 
\begin{equation}\label{pac2_check}
\begin{array}{ll}
(b+c)a+b+c+1\ne0, & (b\gamma+c)a+\gamma(b\gamma+c+1)\ne0,\\
(b+1)a+b+c+1\ne0, & (b\gamma+1)a+\gamma(b\gamma+c+1)\ne0.
\end{array}
\end{equation}
If all conditions in the above theorem hold, then $bx^2+(ab+c+1)x+ac=0$ has two solutions in $\Fbn\setminus P_a$. Since all solutions are not belong to $P_a$ we can say that they are also solutions of $_cD_aF(x)=b$.
\end{proof}

Next we show that $F$ in \eqref{car3} is not APcN when $c\ne 0$. 

\begin{theorem}\label{cdu_bound_thm}
Under the same assumption as in \eqref{car3} we have $3\le {}_c\Delta_F \le 5$. 
\end{theorem}
\begin{proof}
By Theorem \ref{cdu_upper_inv}, $_c\Delta_F \le 5$ is straightforward. Next we show that $_c\Delta_F \ge 3$. We set $b=\ _c D_{a}F(0)=c+a^{-1}$. Then $x=0$ is a solution of $_c D_{a}F(x)=b$. We use Lemma \ref{panot2_lemma} to show that there is $a\in\Fbn$ such that $_c D_{c}F(x)=b$ has two solutions in $\Fbn\setminus P_a$. It is easy to see that $b\ne 0$ since $c\ne a^{-1}$, and $ab+c+1=a(a^{-1}+c)+c+1=c(a+1)\ne 0$ if $a\ne 1$. Now we assume that $a\ne 1$. We require
\begin{align*}
0&=\tr\left(\frac{abc}{(ab+c+1)^2}\right)=\tr\left(\frac{ac(a^{-1}+c)}{c^2(a^2+1)}\right)=\tr\left(\frac{ca+1}{c(a^2+1)}\right)=\tr\left(\frac{c(a+1)+c+1}{c(a^2+1)}\right)\\
&=\tr\left(\frac{1}{a+1}+\frac{1}{(a+1)^2}+\frac{1}{c(a+1)^2}\right)=\tr\left(\frac{1}{c(a+1)^2}\right).
\end{align*}
Using $b+c+1=a^{-1}+1$ and $b\gamma+c=a^{-1}\gamma+c(\gamma+1)$ and $\gamma(b\gamma+c+1)=a^{-1}\gamma^2+c\gamma(\gamma+1)+\gamma$ we check \eqref{pac2_check}
\begin{align*}
(b+c)a+b+c+1&=a^{-1}\cdot a +a^{-1}+1=a^{-1}\\
(b+1)a+b+c+1&=a(a^{-1}+c+1)+a^{-1}+1=a^{-1}\left((c+1)a^2+1\right)
\end{align*}
\begin{align*}
(b\gamma+c)a+\gamma(b\gamma+c+1)&=a(a^{-1}\gamma+c(\gamma+1))+a^{-1}\gamma^2+c\gamma(\gamma+1)+\gamma\\
&=a^{-1}\left(c(\gamma+1)a^2+c\gamma(\gamma+1)a+\gamma^2\right)\\
(b\gamma+1)a+\gamma(b\gamma+c+1)&=a(a^{-1}\gamma+c\gamma+1)+a^{-1}\gamma^2+c\gamma(\gamma+1)+\gamma\\
&=a^{-1}\left((c\gamma+1)a^2+c\gamma(\gamma+1)a+\gamma^2\right)
\end{align*}
We set $S=\{0,\frac{1}{c^{2^{n-1}}+1}\}\cup\{x\in \Fbn : c(\gamma+1) x^2+c\gamma(\gamma+1)x+\gamma^2=0\}\cup\{x\in\Fbn : (c\gamma+1)x^2+c\gamma(\gamma+1)x+\gamma^2=0\}$. By Lemma \ref{panot2_lemma}, $_c D_a F(x)=c+a^{-1}$ has two solutions in $\Fbn\setminus P_a$ if and only if $\tr\left(\frac{1}{c(a+1)^2}\right)=0$ and $a\in\Fbn\setminus (S\cup\{1\})$. Since the map $a \mapsto \frac{1}{c(a+1)^2}$ is an injection from $\Fbn\setminus\{1\}$ to $\Fbnmul$ and $\tr(\cdot)$ is a balanced map, we have $\#\left\{a\in\Fbn\setminus\{1\} : \tr\left(\frac{1}{c(a+1)^2}\right)=0\right\}=2^{n-1}-1$. Since $\# S \le 6$, we have $\#\left(\left\{a\in\Fbn\setminus\{1\} : \tr\left(\frac{1}{c(a+1)^2}\right)=0\right\}\setminus S\right)\ge 2^{n-1}-7\ge 1$ if $n\ge4$. Therefore, there exists $a\in \left\{a\in\Fbn\setminus\{1\} : \tr\left(\frac{1}{c(a+1)^2}\right)=0\right\}\setminus S$ such that $_c D_a F(x)=c+a^{-1}$ has two solutions in $\Fbn\setminus P_a$ and hence $_c\Delta_F(a,c+a^{-1})\ge 3$.
\end{proof}

We summarize in Table \ref{cdu_dist_table} the distribution of $_c\Delta_F$ by computing the number of pairs $(c,\gamma)\in (\Fbn\setminus\F_2)\times (\Fbn\setminus\F_2)$ such that $_c\Delta_F$ equals to each value in $\{3,4,5\}$, where the column sum is $(2^n-2)^2$. We can see that we have $_c \Delta_F =4$ in most cases, and the cases $_c\Delta_F=3$ or $_c\Delta_F=5$ are relatively rare. So we characterize the cases that $_c\Delta_F=3$ or $_c\Delta_F=5$. First we give a characterization for the case  $_c\Delta_F=5$.

\begin{table}
\begin{center}
\begin{tabular}{|c|c|c|c|c|c|}
\hline \backslashbox{$_c\Delta_F$}{$n$} & 4 & 5 & 6 & 7 & 8\\
\hline 3 & 32 & 10 & 28 & 196 & 672\\
\hline 4 & 164 & 820 & 3576 & 15176 & 62880 \\
\hline 5 & 0 & 70 & 240 & 504 & 964 \\
\hline 
\end{tabular}
\caption{Distribution of $_c\Delta_F$ when $4\le n \le 8$.}\label{cdu_dist_table}

\end{center}
\vspace{-0.5cm}
\end{table}

\begin{theorem}\label{cdu5_thm}
Under the same assumption as in \eqref{car3} we have $_c\Delta_F = 5$ if and only if at least one of the following conditions is satisfied :\\
(i) $c\in\left\{\frac{\gamma^3}{(\gamma+1)^2},\frac{(\gamma+1)^2}{\gamma^3}\right\}$ and $\tr\left(\frac{1}{\gamma(\gamma+1)^2}\right)=0$ and $\gamma\not\in\F_4$ and $\gamma^4+\gamma^3+1\ne0$ and $\gamma^5+\gamma^2+1\ne0$.\\
(ii) $c\in\left\{\frac{\gamma+1}{\gamma^3+\gamma^2+\gamma},\frac{\gamma^3+\gamma^2+\gamma}{\gamma+1}\right\}$ and $\tr\left(1+\frac{1}{\gamma^3}\right)=0$ and $\gamma\not\in\F_4$ and $\gamma^5+\gamma^3+1\ne0$.\\
(iii) $c\in\left\{\frac{\gamma^{2^{n-1}+1}+\gamma^{2^{n-1}}+\gamma}{\gamma(\gamma+1)},\frac{\gamma(\gamma+1)}{\gamma^{2^{n-1}+1}+\gamma^{2^{n-1}}+\gamma}\right\}$ and $\tr\left(\frac{\gamma+1}{\gamma^2(\gamma^2+\gamma+1)}\right)=0$ and $\gamma\not\in\F_4$ and $\gamma^5\ne1$.\\
(iv) $n\equiv 0\pmod{8}$ and $c,\gamma\in\F_4 \setminus \F_2$.\\
(v) $c\in\{x\in\Fbn : \gamma^3 x^3+\gamma^2 x^2+(\gamma+1)x+1=0\}\setminus\{\gamma,\frac{1}{\gamma^{2^{n-1}}+\gamma}\}$, $\gamma^3c^2+(\gamma^2+\gamma+1)c+\gamma^2\ne0$ and $\tr\left(\frac{(c\gamma+c+1)(c^2\gamma+1)}{c(c+\gamma)^2}\right)=0$ and $\gamma\not\in\F_4$.\\
(vi) $c\in\{x\in\Fbn : x^3+(\gamma+1) x^2+\gamma^2 x+\gamma^3=0\}\setminus\{\gamma^{-1}, \gamma^{2^{n-1}}+\gamma\}$, $\gamma^2c^2+(\gamma^2+\gamma+1)c+\gamma^3\ne0$. and $\tr\left(\frac{(c+\gamma+1)(c^2+\gamma)}{(c\gamma+1)^2}\right)=0$ and $\gamma\not\in\F_4$.
\end{theorem}
\begin{proof}
Assume that $_c\Delta_F=5$. When $a\in\{0,1,\gamma,\gamma+1\}$, we have $\#P_a=4$ and hence $_c\Delta_F(a,b)\le 4$ by \eqref{car3cdu_eqn}. Thus we assume that $a\in \Fbn \setminus \{0,1,\gamma,\gamma+1\}$, and then we have $\# P_a=6$. We can see in the proof of Theorem \ref{cdu_upper_inv} that $_cD_a(x)=b$ has at most $3$ solutions in $P_a$ and at most $2$ solutions in $\Fbn\setminus P_a$. Since $_c\Delta_F = 5$, there are $a\in\Fbn\setminus\{0,1,\gamma,\gamma+1\}$ and $b\in\Fbn$ such that $_cD_aF(x)=b$ has $3$ solutions in $P_a$ and $2$ solutions in $\Fbn\setminus P_a$. First we consider that $_cD_aF(x)=b$ has $3$ solutions in $P_a$. We state $_cD_a F(x)=F(x+a)+cF(x)$ for each $x\in P_a$.
\begin{equation*}
\begin{array}{ll}
_cD_a F(0)=F(a)+cF(0)=a^{-1}+c & _cD_a F(a)=F(0)+cF(a)=1+ca^{-1}\\
_cD_a F(1)=F(1+a)+cF(1)=(a+1)^{-1}+\gamma^{-1}c & _cD_a F(a+1)=F(1)+cF(a+1)=\gamma^{-1}+c(a+1)^{-1}\\
_cD_a F(\gamma)=F(\gamma+a)+cF(\gamma)=(a+\gamma)^{-1} & _cD_a F(a+\gamma)=F(\gamma)+cF(\gamma+a)=c(a+\gamma)^{-1}
\end{array}
\end{equation*}
Applying \eqref{asa_eqn}, we investigate required conditions such that $_cD_aF(x)=b$ has at least two solutions in $P_a$.
\begin{itemize}
\item $_c D_a F(0)={}_c D_a F(1)\ \Rightarrow\ c(\gamma+1) a^2+c(\gamma+1) a+\gamma=0$ (has no solutions $\Leftrightarrow \tr\left(\frac{\gamma}{c(\gamma+1)}\right)=1$). 
\item $_c D_a F(0)={}_c D_a F(\gamma)\ \Rightarrow\ c a^2+c\gamma a+\gamma=0$ (has no solutions $\Leftrightarrow \tr(c^{-1}\gamma^{-1})=1$). 
\item $_c D_a F(0)={}_c D_a F(a+1)\ \Rightarrow\ (c\gamma+1) a^2+(\gamma+1) a+\gamma=0$ (has no solutions $\Leftrightarrow \tr\left(\frac{\gamma(c\gamma+1)}{(\gamma+1)^2}\right)=1$ and $c\ne\gamma^{-1}$).
\item  $_c D_a F(0)={}_c D_a F(a+\gamma)\ \Rightarrow\ c a^2+(c\gamma+c+1)a+\gamma=0$ (has no solutions $\Leftrightarrow \tr\left(\frac{c\gamma}{(c\gamma+c+1)^2}\right)=1$ and $c\ne (\gamma+1)^{-1}$).
\item $_c D_a F(1)={}_c D_a F(\gamma)\ \Rightarrow\ c a^2+c(\gamma+1) a+\gamma(c+\gamma+1)=1$ (has no solutions $\Leftrightarrow \tr\left(\frac{\gamma(c+\gamma+1)}{c(\gamma+1)^2}\right)=1$ and $c\ne \gamma+1$).
\item $_c D_a F(1)={}_c D_a F(a)\ \Rightarrow\ (c+\gamma) a^2+c(\gamma+1) a+c\gamma=0$ (has no solutions $\Leftrightarrow \tr\left(\frac{\gamma(c+\gamma)}{c(\gamma+1)^2}\right)=1$ and $c\ne \gamma$).
\item $_c D_a F(1)={}_c D_a F(a+\gamma)\ \Rightarrow\ c a^2+(c+\gamma) a+\gamma^2=0$ (has no solutions $\Leftrightarrow \tr\left(\frac{c\gamma^2}{(c+\gamma)^2}\right)=1$ and $c\ne \gamma$).
\item $_c D_a F(\gamma)={}_c D_a F(a)\ \Rightarrow\ a^2+(c+\gamma+1) a+c\gamma=0$ (has no solutions $\Leftrightarrow \tr\left(\frac{c\gamma}{(c+\gamma+1)^2} \right)=1$ and $c\ne \gamma+1$).
\item $_c D_a F(\gamma)={}_c D_a F(a+1)\ \Rightarrow\  a^2+(c\gamma+1) a+c\gamma^2=0$ (has no solutions $\Leftrightarrow \tr\left(\frac{c\gamma^2}{(c\gamma+1)^2} \right)=1$ and $c\ne \gamma^{-1}$).
\item $_c D_a F(a)={}_c D_a F(a+1)\ \Rightarrow\ (\gamma+1) a^2+(\gamma+1) a+c\gamma=0$ (has no solutions $\Leftrightarrow \tr\left(\frac{c\gamma}{\gamma+1}\right)=1$).
\item $_c D_a F(a)={}_c D_a F(a+\gamma)\ \Rightarrow\ a^2+\gamma a+c\gamma=0$ (has no solutions $\Leftrightarrow \tr(c\gamma^{-1})=1$).
\item $_c D_a F(a+1)={}_c D_a F(a+\gamma)\ \Rightarrow\  a^2+(\gamma+1) a+(c\gamma+c+1)\gamma=0$ (has no solutions $\Leftrightarrow \tr\left(\frac{(c\gamma+c+1)\gamma}{(\gamma+1)^2}\right)=1$ and $c\ne (\gamma+1)^{-1}$).
\end{itemize}

Next we investigate all the cases that $_cD_aF(x)=b$ has $3$ solutions in $P_a$. For each case, we apply Lemma \ref{panot2_lemma} to find required conditions that $_cD_aF(x)=b$ has $2$ solutions in $\Fbn\setminus P_a$.

\bigskip 
\noindent (Case 1) If $_c D_a F(0)= {}_c D_a F(1) = {}_c D_a F(\gamma)$  then $a$ is a common solution of $c(\gamma+1) a^2+c(\gamma+1) a+\gamma=0$, $c a^2+c\gamma a+\gamma=0$ and $ca^2+c(\gamma+1) a+\gamma(c+\gamma+1)=0$. We add the first two equations to get $0=c\gamma a^2+ca=ca(\gamma a+1)$ and hence we have $a=\gamma^{-1}$. We have $c=\frac{\gamma^3}{(\gamma+1)^2}$ to substitute $a=\gamma^{-1}$ to each equation. We set $b= {}_cD_a(\gamma)=\frac{\gamma}{\gamma^2+1}$ and check Lemma \ref{panot2_lemma}. Then we have $b\ne0$ and 
$$ab+c+1=\frac{1}{\gamma^2+1}+\frac{\gamma^3}{\gamma^2+1}+1=\frac{\gamma^3+\gamma^2}{\gamma^2+1}=\frac{\gamma^2}{\gamma+1}\ne0,$$
since $\gamma\ne0$. And \eqref{cdu5_trace} becomes 
$$0=\tr\left(\frac{abc}{(ab+c+1)^2}\right)=\tr\left(\frac{\frac{\gamma^3}{(\gamma+1)^4}}{\frac{\gamma^4}{(\gamma+1)^2}}\right)=\tr\left(\frac{1}{\gamma(\gamma+1)^2}\right).$$
Using $b+c+1=\frac{\gamma^3+\gamma}{\gamma^2+1}+1=\gamma+1$ and $\gamma(b\gamma+c+1)=\frac{\gamma(\gamma^3+1)}{(\gamma+1)^2}=\frac{\gamma(\gamma^2+\gamma+1)}{\gamma+1}$, \eqref{pac2_check} becomes
\begin{align*}
a(b+c)+b+c+1&=\gamma^{-1}\cdot \gamma+\gamma+1=\gamma\ne 0\\
a(b+1)+b+c+1&=\gamma^{-1}\cdot \frac{\gamma^2+\gamma+1}{\gamma^2+1}+\gamma+1=\frac{(\gamma^2+\gamma+1)+(\gamma^4+\gamma^3+\gamma^2+\gamma)}{\gamma(\gamma+1)^2}=\frac{\gamma^4+\gamma^3+1}{\gamma(\gamma+1)^2}\\
(b\gamma+c)a+\gamma(b\gamma+c+1)&=\gamma^{-1}\cdot \frac{\gamma^3+\gamma^2}{(\gamma+1)^2}+\frac{\gamma(\gamma^2+\gamma+1)}{\gamma+1}=\frac{\gamma^2(\gamma+1)}{\gamma+1}=\gamma^2\ne0\\
(b\gamma+1)a+\gamma(b\gamma+c+1)&=\gamma^{-1}\cdot \frac{1}{(\gamma+1)^2}+\frac{\gamma(\gamma^2+\gamma+1)}{\gamma+1}=\frac{\gamma^2(\gamma^3+1)+1}{\gamma(\gamma+1)^2}=\frac{\gamma^5+\gamma^2+1}{\gamma(\gamma+1)^2}
\end{align*}
Hence we require $\gamma^4+\gamma^3+1\ne0$ and $\gamma^5+\gamma^2+1\ne0$. It is easy to check $a=\gamma^{-1}\not\in\{0,1,\gamma,\gamma+1\}$ requires $\gamma\not\in\F_4$.\\
If $c=\frac{(\gamma+1)^2}{\gamma^3}$, then we have $_c\Delta_F= {}_{c^{-1}}\Delta_F=5$ by Lemma \ref{cdu_symm_lemma}. We exchange $c$ by $c^{-1}$ in the above analysis, we also require $\tr\left(\frac{1}{\gamma(\gamma+1)^2}\right)=0$, $\gamma\not\in\F_4$, $\gamma^4+\gamma^3+1\ne0$ and $\gamma^5+\gamma^2+1\ne0$.

\bigskip
\noindent (Case 2) If $_c D_a F(0)= {}_c D_a F(1) = {}_c D_a F(a+\gamma)$  then $a$ is a common solution of $c(\gamma+1) a^2+c(\gamma+1) a+\gamma=0$, $c a^2+(c\gamma+c+1)a+\gamma=0$ and $c a^2+(c+\gamma) a+\gamma^2=0$. We add the first two equation to have $c\gamma a^2+a=0$ and hence $ca=\gamma^{-1}$. Hence the first equation implies $0=ca(\gamma+1)a +ca(\gamma+1) +\gamma=\frac{(\gamma+1)a+(\gamma^2+\gamma+1)}{\gamma}$ and hence we have $a=\frac{\gamma^2+\gamma+1}{\gamma+1}$ and $c=\frac{1}{a\gamma}=\frac{\gamma+1}{\gamma^3+\gamma^2+\gamma}$. Note that we have $a=0$ if $\gamma\in\F_4\setminus \F_2$ and hence we require $\gamma\not\in\F_4$.
We set $b= {} _c D_a F(0)=a^{-1}+c=\frac{\gamma^2+1}{\gamma^3+\gamma^2+\gamma}$ and check Lemma \ref{panot2_lemma}. Then we have $b\ne0$ and 
\begin{align*}
 ab+c+1&=\frac{\gamma^2+\gamma+1}{\gamma+1}\cdot \frac{\gamma^2+1}{\gamma^3+\gamma^2+\gamma}+\frac{\gamma+1}{\gamma^3+\gamma^2+\gamma}+1= \frac{(\gamma^3+1)+(\gamma+1)+(\gamma^3+\gamma^2+\gamma)}{\gamma^3+\gamma^2+\gamma}\\
&=\frac{\gamma^2}{\gamma^3+\gamma^2+\gamma}\ne0
\end{align*} 
 since $\gamma\ne0,1$. And \eqref{cdu5_trace} becomes
$$0=\tr\left(\frac{abc}{(ab+c+1)^2}\right)=\tr\left(\frac{\frac{(\gamma+1)^2}{\gamma^2(\gamma^2+\gamma+1)}}{\frac{\gamma^2}{(\gamma^2+\gamma+1)^2}}\right)=\tr\left(\frac{\gamma^4+\gamma^3+\gamma+1}{\gamma^4}\right)=\tr\left(1+\frac{1}{\gamma^3}\right).$$
Using $b+c+1=\frac{\gamma+1}{\gamma^2+\gamma+1}+1=\frac{\gamma^2}{\gamma^2+\gamma+1}$ and $\gamma(b\gamma+c+1)=\gamma\left(\frac{\gamma^3+1}{\gamma(\gamma^2+\gamma+1)}+1\right)=\gamma\left(\frac{\gamma+1}{\gamma}+1\right)=1$, \eqref{pac2_check} becomes
\begin{align*}
a(b+c)+b+c+1&=\frac{\gamma^2+\gamma+1}{\gamma+1}\cdot\frac{\gamma+1}{\gamma^2+\gamma+1}+\frac{\gamma+1}{\gamma^2+\gamma+1}+1=\frac{\gamma+1}{\gamma^2+\gamma+1}\ne0\\
a(b+1)+b+c+1&=\frac{\gamma^2+\gamma+1}{\gamma+1}\cdot \frac{\gamma^3+\gamma+1}{\gamma(\gamma^2+\gamma+1)}+\frac{\gamma^2}{\gamma^2+\gamma+1}=\frac{\gamma^5+\gamma^3+1}{\gamma(\gamma^3+1)}\\
(b\gamma+c)a+\gamma(b\gamma+c+1)&=\frac{\gamma^2+\gamma+1}{\gamma+1}\cdot \frac{\gamma+1}{\gamma}+1=\frac{\gamma^2+1}{\gamma}\ne0\\
(b\gamma+1)a+\gamma(b\gamma+c+1)&=\frac{\gamma^2+\gamma+1}{\gamma+1}\cdot \frac{\gamma}{\gamma^2+\gamma+1}+1=\frac{1}{\gamma+1}\ne0
\end{align*}
Hence we require $\gamma^5+\gamma^3+1\ne0$. It is easy to check $a=\frac{\gamma^2+\gamma+1}{\gamma+1}\not\in\{0,1,\gamma,\gamma+1\}$ if $\gamma\not\in\F_4$. \\
If $c=\frac{\gamma^3+\gamma^2+\gamma}{\gamma+1}$, then we have $_c\Delta_F= {}_{c^{-1}}\Delta_F=5$ by Lemma \ref{cdu_symm_lemma}. We exchange $c$ by $c^{-1}$ in the above analysis, we also require $\tr\left(1+\frac{1}{\gamma^3}\right)=0$ and $\gamma\not\in\F_4$ and $\gamma^5+\gamma^3+1\ne0$.

\bigskip 
\noindent (Case 3) If $_c D_a F(0)= {}_c D_a F(a+1) = {}_c D_a F(a+\gamma)$ then $a$ is a common solution of $(c\gamma+1) a^2+(\gamma+1) a+\gamma=0$, $ c a^2+(c\gamma+c+1)a+\gamma=0$ and $a^2+(\gamma+1) a+(c\gamma+c+1)\gamma=0$. We add the first equation and the third equation to get $c\gamma a^2 +c\gamma(\gamma+1)=0$ and hence $a^2=\gamma+1$ so $a=\gamma^{2^{n-1}}+1$. We substitute $a=\gamma^{2^{n-1}}+1$ to the third equation to have $0=c\gamma(\gamma+1)+\gamma^{2^{n-1}+1}+\gamma^{2^{n-1}}+\gamma$ and hence $c=\frac{\gamma^{2^{n-1}+1}+\gamma^{2^{n-1}}+\gamma}{\gamma(\gamma+1)}\ne 0$ because $0\ne \gamma(\gamma^2+\gamma+1)=(\gamma^{2^{n-1}+1}+\gamma^{2^{n-1}}+\gamma)^2$ when $\gamma\not\in\F_4$. We set $b= {}_c D_a F(a+\gamma)=\frac{\gamma^{2^{n-1}}}{\gamma(\gamma+1)}$ and check Lemma \ref{panot2_lemma}. Then we have $b\ne0$ and
$$ab+c+1=\frac{\gamma^{2^{n-1}}(\gamma^{2^{n-1}}+1)+(\gamma^{2^{n-1}+1}+\gamma^{2^{n-1}}+\gamma)+\gamma^2+\gamma}{\gamma(\gamma+1)}=\frac{\gamma^{2^{n-1}+1}+\gamma^2+\gamma}{\gamma(\gamma+1)}\ne0$$
requires $\gamma\not\in\F_4\setminus \F_2$ because $0\ne \gamma(\gamma^2+\gamma+1)=(\gamma^{2^{n-1}+1}+\gamma^{2^{n-1}}+\gamma)^2$. And \eqref{cdu5_trace} becomes
\begin{align*}
0&=\tr\left(\frac{abc}{(ab+c+1)^2}\right)=\tr\left(\frac{\frac{\gamma^{2^{n-1}}(\gamma^{2^{n-1}}+1)(\gamma^{2^{n-1}+1}+\gamma^{2^{n-1}}+\gamma)}{\gamma^2(\gamma+1)^2}}{\frac{\gamma^4+\gamma^3+\gamma^2}{\gamma^2(\gamma+1)^2}}\right)\\
&=\tr\left(\frac{\gamma^{2^{n-1}}(\gamma^{2^{n-1}}+1)(\gamma^{2^{n-1}+1}+\gamma^{2^{n-1}}+\gamma)}{\gamma^2(\gamma^2+\gamma+1)}\right)=\tr\left(\frac{\gamma(\gamma+1)(\gamma^{3}+\gamma^{2}+\gamma)}{\gamma^4(\gamma^2+\gamma+1)^2}\right)\\
&=\tr\left(\frac{\gamma+1}{\gamma^2(\gamma^2+\gamma+1)}\right).
\end{align*}
Using $b+c+1=\frac{\gamma^{2^{n-1}}+1}{\gamma+1}+1=\frac{\gamma^{2^{n-1}}+\gamma}{\gamma+1}$ and $\gamma(b\gamma+c+1)=\gamma\left(\frac{\gamma^{2^{n-1}}+\gamma}{\gamma(\gamma+1)}+1\right)=\frac{\gamma^{2^{n-1}+1}+\gamma^3}{\gamma(\gamma+1)}$, \eqref{pac2_check} becomes
\begin{align*}
a(b+c)+b+c+1&=(\gamma^{2^{n-1}}+1)\cdot\frac{\gamma^{2^{n-1}}+1}{\gamma+1}+\frac{\gamma^{2^{n-1}}+\gamma}{\gamma+1}=\frac{\gamma^{2^{n-1}}+1}{\gamma+1}=(\gamma+1)^{2^{n-1}-1}\ne0 \\
a(b+1)+b+c+1&=(\gamma^{2^{n-1}}+1)\cdot \frac{\gamma^{2^{n-1}}+\gamma^2+\gamma}{\gamma(\gamma+1)}+\frac{\gamma^{2^{n-1}}+\gamma}{\gamma+1}=\gamma^{2^{n-1}-1}(\gamma+1)\ne0\\
(b\gamma+c)a+\gamma(b\gamma+c+1)&=(\gamma^{2^{n-1}}+1)\cdot \frac{\gamma^{2^{n-1}}+\gamma}{\gamma(\gamma+1)}+\frac{\gamma^{2^{n-1}+1}+\gamma^3}{\gamma(\gamma+1)}=\frac{\gamma^{2^{n-1}}+\gamma^3}{\gamma(\gamma+1)}=\frac{(\gamma^6+\gamma)^{2^{n-1}}}{\gamma(\gamma+1)}\\
(b\gamma+1)a+\gamma(b\gamma+c+1)&=(\gamma^{2^{n-1}}+1)\cdot \frac{\gamma^{2^{n-1}+1}+\gamma^2+\gamma}{\gamma(\gamma+1)}+\frac{\gamma^{2^{n-1}+1}+\gamma^3}{\gamma(\gamma+1)}=\gamma^{2^{n-1}}+\gamma+1\\
&=(\gamma^2+\gamma+1)^{2^{n-1}}\ne0
\end{align*}
Hence we require $\gamma^5\ne1$. It is easy to check $a=\gamma^{2^{n-1}}+1\not\in\{0,1,\gamma,\gamma+1\}$ if $\gamma\not\in\F_4$.  \\
If $c=\frac{\gamma(\gamma+1)}{\gamma^{2^{n-1}+1}+\gamma^{2^{n-1}}+\gamma}$, then we have $_c\Delta_F={}_{c^{-1}}\Delta_F=5$ by Lemma \ref{cdu_symm_lemma}. We exchange $c$ by $c^{-1}$ in the above analysis, we also require $\tr\left(\frac{\gamma+1}{\gamma^2(\gamma^2+\gamma+1)}\right)=0$, $\gamma\not\in\F_4$ and $\gamma^5\ne1$.

\bigskip
\noindent (Case 4) If $_c D_a F(0)= {}_c D_a F(\gamma) = {}_c D_a F(a+1)$ then $a$ is a common solution of $c a^2+c\gamma a+\gamma=0$, $(c\gamma+1) a^2+(\gamma+1) a+\gamma=0$ and $a^2+(c\gamma+1) a+c\gamma^2=0$. We add the first two equations to obtain $(c\gamma+c+1)a^2+(c\gamma+\gamma+1)a=0$. Hence we have $c\gamma+c+1=c\gamma+\gamma+1=0$ or $a=\frac{c\gamma+\gamma+1}{c\gamma+c+1}$. 
\begin{itemize}
\item If $c\gamma+c+1=c\gamma+\gamma+1=0$ then $c=\gamma$ and $\gamma^2+\gamma+1=0$ so $\gamma\in\F_4$. And $a$ is a solution of $a^2+\gamma a +\gamma=0$ where such $a$ exists if and only if $\tr(\gamma)=0$ if and only if $n\equiv 0\pmod{4}$. Since $a^2=a\gamma+1$, we have
\begin{align*}
a^4&=a^2 \gamma^2+1=(a\gamma+1)\gamma^2+1=a+\gamma,\\
a^8&=a^2+\gamma^2=a\gamma+1+\gamma^2=a\gamma+\gamma.
\end{align*}
Thus $a+a^2+a^4+a^8=a+(a\gamma+1)+(a+\gamma)+(a\gamma+\gamma)=1$ and hence
\begin{equation*}
\tr(a)=\sum_{i=0}^{n-1} a^{2^i}=\sum_{j=0}^{\frac{n}{4}}(a+a^2+a^{2^2}+a^{2^3})^{2^j}=\sum_{j=0}^{\frac{n}{4}}1=\frac{n}{4}
\end{equation*}
 We set $b= {}_c D_a F(0)= a^{-1}+\gamma$ and check Lemma \ref{panot2_lemma}. Then we have $b\ne 0$, since if $a=\gamma^{-1}=\gamma^2$ then $0=a^2+\gamma a +1=\gamma$ which is a contradiction. We also have $ab+c+1=a(a^{-1}+\gamma)+\gamma+1=\gamma(a+1)\ne0$, since if $a=1$ then $0=a^2+\gamma a +1=\gamma$ which is a contradiction. And \eqref{cdu5_trace} becomes
\begin{align*}
0&=\tr\left(\frac{abc}{(ab+c+1)^2}\right)=\tr\left(\frac{a(a^{-1}+\gamma)\gamma}{\gamma^2(a+1)^2}\right)=\tr\left(\frac{(a\gamma+1)\gamma}{a\gamma^3}\right)=\tr\left(\frac{a^2\gamma}{a}\right)\\
&=\tr(a\gamma)=\tr(a^2+1)=\tr(a)
\end{align*}
if and only if $n\equiv 0\pmod{8}$. We can easily check that $a\not\in\{0,1,\gamma,\gamma+1\}=\F_4$ if $\gamma\in\F_4\setminus\F_2$ and $a^2+a\gamma+1=0$. 
Using $b+c+1=a^{-1}+1$ and $\gamma(b\gamma+c+1)=a^{-1}\gamma^2+\gamma^3+\gamma^2+\gamma=a^{-1}\gamma^2$, \eqref{pac2_check} becomes
\begin{align*}
a(b+c)+b+c+1&=a\cdot a^{-1}+a^{-1}+1=a^{-1}\ne0\\
a(b+1)+b+c+1&=a(a^{-1}+\gamma+1) +a^{-1}+1=a\gamma+a+a^{-1}=a^{-1}\gamma^2(a^2+\gamma)\ne0\\
(b\gamma+c)a+\gamma(b\gamma+c+1)&=a(a^{-1}\gamma+1)+a^{-1}\gamma^2=a^{-1}(a^2+\gamma a+\gamma^2)=a^{-1}\gamma \ne0\\
(b\gamma+1)a+\gamma(b\gamma+c+1)&=a\gamma(a^{-1}+1)+a^{-1}\gamma^2=a^{-1}\gamma(a^2+a+\gamma)=a^{-1}( a+1)\ne0
\end{align*}
Note that $_{\gamma^2}\Delta_F= {}_{\gamma^{-1}}\Delta_F= {}_{\gamma}\Delta_F=5$ in this case.
\item Now we assume that $\gamma\in\Fbn\setminus\F_4$ and $a=\frac{c\gamma+\gamma+1}{c\gamma+c+1}$. We substitute $a=\frac{c\gamma+\gamma+1}{c\gamma+c+1}$ to the first equation to get 
\begin{align*}
0&=c(c\gamma+\gamma+1)^2+c\gamma(c\gamma+c+1)(c\gamma+\gamma+1)+\gamma(c\gamma+c+1)^2\\
&=\gamma^3 c^3 +\gamma^2 c^2+(\gamma+1)c+\gamma
\end{align*}
We have the same results when we substitute $a=\frac{c\gamma+\gamma+1}{c\gamma+c+1}$ to the last two equations. We set $b=\  _c D_a F(0)=\frac{c^2\gamma+1}{c\gamma+\gamma+1}$ check Lemma \ref{panot2_lemma}. If $c^2\gamma=1$ then $=\gamma^3 c^3 +\gamma^2 c^2+(\gamma+1)c+\gamma=(\gamma^2+\gamma+1)c$, which is a contradiction to $c\ne0$ and $\gamma\not\in\F_4$. Hence we have $b\ne 0$. Moreover,
$$ab+c+1=\frac{c^2\gamma+1+(c+1)(c\gamma+c+1)}{c\gamma+c+1}=\frac{c(c+\gamma)}{c\gamma+c+1}\ne0$$
if and only if $c\ne \gamma$. And \eqref{cdu5_trace} becomes
\begin{align*}
0&=\tr\left(\frac{abc}{(ab+c+1)^2}\right)=\tr\left(\frac{\frac{c(c^2\gamma+1)}{c\gamma+c+1}}{\frac{c^2(c+\gamma)^2}{(c\gamma+c+1)^2}}\right)=\tr\left(\frac{(c\gamma+c+1)(c^2\gamma+1)}{c(c+\gamma)^2}\right).
\end{align*}
Using $b+c=\frac{c\gamma+c+1}{c\gamma+\gamma+1}=a^{-1}$ and hence $b+c+1=\frac{c+\gamma}{c\gamma+\gamma+1}$ and $\gamma(b\gamma+c+1)=\frac{c^2\gamma^2(\gamma+1)+c\gamma+\gamma}{c\gamma+\gamma+1}$, \eqref{pac2_check} becomes
\begin{align*}
a(b+c)+b+c+1&=a\cdot a^{-1}+a^{-1}+1=a^{-1}\ne0\\
a(b+1)+b+c+1&=\frac{c\gamma+\gamma+1}{c\gamma+c+1}\cdot \frac{\gamma(c^2+c+1)}{c\gamma+\gamma+1}+\frac{c+\gamma}{c\gamma+\gamma+1}=\frac{\gamma^2c^3+c^2+(\gamma+1)^2c+\gamma^2}{(c\gamma+c+1)(c\gamma+\gamma+1)}\\
&=\frac{\gamma^3c^3+\gamma c^2+\gamma(\gamma+1)^2c+\gamma^3}{\gamma(c\gamma+c+1)(c\gamma+\gamma+1)}=\frac{(\gamma^2+\gamma) c^2+(\gamma^3+1)c+\gamma^3+\gamma}{\gamma(c\gamma+c+1)(c\gamma+\gamma+1)}\\
&=\frac{(\gamma+1)( c+\gamma)(\gamma c+\gamma+1)}{\gamma(c\gamma+c+1)(c\gamma+\gamma+1)}=\frac{(\gamma+1)( c+\gamma)}{\gamma(c\gamma+c+1)}\ne 0
\\
(b\gamma+c)a+\gamma(b\gamma+c+1)&=\frac{c\gamma+\gamma+1}{c\gamma+c+1}\cdot \frac{c^2\gamma(\gamma+1)+c(\gamma+1)+\gamma}{c\gamma+\gamma+1}+\frac{c^2\gamma^2(\gamma+1)+c\gamma+\gamma}{c\gamma+\gamma+1}\\
&=\frac{(\gamma+1)\gamma^3c^3+(\gamma+1)\gamma c^2+(\gamma+1)^2 c+\gamma^2}{(c\gamma+c+1)(c\gamma+\gamma+1)}\\
&=\frac{(\gamma+1)(\gamma^3c^3+\gamma^2 c^2+(\gamma+1) c)+\gamma^2(\gamma+1)c^2+\gamma^2}{(c\gamma+c+1)(c\gamma+\gamma+1)}\\
&=\frac{\gamma^2(\gamma+1)c^2+\gamma}{(c\gamma+c+1)(c\gamma+\gamma+1)}=\frac{\gamma(\gamma(\gamma+1)c^2+1)}{(c\gamma+c+1)(c\gamma+\gamma+1)}
\end{align*}
\begin{align*}
(b\gamma+1)a+\gamma(b\gamma+c+1)&=\frac{c\gamma+\gamma+1}{c\gamma+c+1}\cdot \frac{c^2\gamma^2+c\gamma+1}{c\gamma+\gamma+1}+\frac{c^2\gamma^2(\gamma+1)+c\gamma+\gamma}{c\gamma+\gamma+1}\\
&=\frac{(\gamma^2+\gamma+1)\gamma^2c^3+\gamma c^2+1}{(c\gamma+c+1)(c\gamma+\gamma+1)}=\frac{(\gamma^2+\gamma+1)\gamma^3c^3+\gamma^2 c^2+\gamma}{\gamma(c\gamma+c+1)(c\gamma+\gamma+1)}\\
&=\frac{(\gamma^2+\gamma+1)(\gamma^2c^2+(\gamma+1)c+\gamma)+\gamma^2 c^2+\gamma}{\gamma(c\gamma+c+1)(c\gamma+\gamma+1)}\\
&=\frac{(\gamma+1)\gamma^3c^2+(\gamma+1)(\gamma^2+\gamma+1)c+\gamma^2(\gamma+1)}{\gamma(c\gamma+c+1)(c\gamma+\gamma+1)}\\
&=\frac{(\gamma+1)(\gamma^3c^2+(\gamma^2+\gamma+1)c+\gamma^2)}{\gamma(c\gamma+c+1)(c\gamma+\gamma+1)}
\end{align*}
Hence we require $c\ne \frac{1}{\gamma^{2^{n-1}}+\gamma}$ and $\gamma^3c^2+(\gamma^2+\gamma+1)c+\gamma^2\ne0$. It remains to check $a=\frac{c\gamma+\gamma+1}{c\gamma+c+1}\not\in\{0,1,\gamma,\gamma+1\}$. If $a=0$ then we have $c\gamma=\gamma+1$ or $c=\frac{\gamma+1}{\gamma}$ and hence we have $0=(c\gamma)^3+(c\gamma)^2+(\gamma+1)c+\gamma=\frac{(\gamma^2+\gamma+1)^2}{\gamma}$ which is a contradiction to $\gamma\not\in\F_4$. If $a=1$ then we have $c=\gamma$ and then we have $ab+c+1\ne0$. If $a=\gamma$ then we have $0=(c\gamma+\gamma+1)+\gamma(c\gamma+c+1)=c\gamma^2+1$, but we already see that $c\gamma^2+1\ne0$. If $a=\gamma+1$ then we have  $0=(c\gamma+\gamma+1)+(\gamma+1)(c\gamma+c+1)=c(\gamma^2+\gamma+1)$ which is a contradiction to $\gamma\not\in\F_4$.
\item We exchange $c$ by $c^{-1}$ in the above analysis to have $a=\frac{c^{-1}\gamma+\gamma+1}{c^{-1}\gamma+c^{-1}+1}=\frac{c\gamma+c+\gamma}{c+\gamma+1}$ and $b=\frac{(c^{-1})^2\gamma+1}{c^{-1}\gamma+\gamma+1}=\frac{c^2+\gamma}{c^2(\gamma+1)+c\gamma}$. Then $ _{c^{-1}} D_aF(x)=b$ has three solutions in $P_a$ if $x=c^{-1}$ is a solution of $\gamma^3x^3+\gamma^2x^2+(\gamma+1)x+1=0$ and hence $c^3+(\gamma+1)c^2+\gamma^2 c+\gamma^3=0$. Similarly, by Lemma \ref{panot2_lemma}, $ _{c^{-1}} D_aF(x)=b$ has two solutions in $\Fbn\setminus P_a$ if and only if 
\begin{align*}
0=\tr\left(\frac{(c^{-1}\gamma+c^{-1}+1)((c^{-1})^2\gamma+1)}{c^{-1}(c^{-1}+\gamma)^2}\right)=\tr\left(\frac{(\gamma+c+1)(c^2+\gamma)}{(c\gamma+1)^2}\right).
\end{align*}
and  $c\ne \gamma^{2^{n-1}}+\gamma$ and $\gamma^2c^2+(\gamma^2+\gamma+1)c+\gamma^3\ne0$.
\end{itemize}

By Lemma \ref{cdu_symm_lemma}, $b={}_c D_a F(u_1)={}_c D_a F(u_2) ={}_c D_a F(u_3)$ if and only if $bc^{-1}={}_{c^{-1}} D_a F(u_1+a)={}_{c^{-1}} D_a F(u_2+a) ={}_{c^{-1}} D_a F(u_3+a)$ for all $u_1, u_2, u_3\in \Fbn$, and hence it is enough to consider the above cases. Therefore, if conditions in this theorem are not satisfied, then $_cD_a=b$ has at most $2$ solutions in $P_a$ or at most $1$ solution in $\Fbn\setminus P_a$ and hence we have $_c\Delta_F\le 4$.

Conversely, if each condition in this theorem holds, then we set $a$ and $b$ the same as in the above analysis in each case. By the above analysis in each case, we can see that $_cD_aF(x)=b$ has $3$ solutions in $P_a$. We can also see that $_cD_aF(x)=b$ has $2$ solutions in $\Fbn\setminus P_a$ by Lemma \ref{panot2_lemma}. Therefore we have $_c\Delta_F(a,b)=5$ and hence $_c\Delta_F=5$ by Theorem \ref{cdu_bound_thm}, which completes the proof.
\end{proof}

By Theorem \ref{cdu5_thm}, we can say that $3\le {}_c\Delta_F \le 4$ if and only if all the conditions in Theorem \ref{cdu5_thm} do not hold. Next we give a simple characterization for the case $_c\Delta_F=3$.

\begin{corollary}\label{cdu3_coro} We have $_c\Delta_F=3$ if $c\not\in\{\gamma,\gamma^{-1},\gamma+1,(\gamma+1)^{-1}\}$ and $\tr\left(\frac{\gamma}{c(\gamma+1)}\right)=
\tr(c^{-1}\gamma^{-1})=
\tr\left(\frac{\gamma(c\gamma+1)}{(\gamma+1)^2}\right)=
\tr\left(\frac{c\gamma}{(c\gamma+c+1)^2}\right)=
\tr\left(\frac{\gamma(c+\gamma+1)}{c(\gamma+1)^2}\right)=
\tr\left(\frac{\gamma(c+\gamma)}{c(\gamma+1)^2}\right)=
\tr\left(\frac{c\gamma^2}{(c+\gamma)^2}\right)=
\tr\left(\frac{c\gamma}{(c+\gamma+1)^2} \right)=
\tr\left(\frac{c\gamma^2}{(c\gamma+1)^2} \right)=
\tr\left(\frac{c\gamma}{\gamma+1}\right)=
\tr(c\gamma^{-1})=
\tr\left(\frac{(c\gamma+c+1)\gamma}{(\gamma+1)^2}\right)
=1$.
\end{corollary}

\begin{proof}
We can see in the proof of Theorem \ref{cdu5_thm} that $_c D_a F(u_1)\ne {}_c D_a F(u_2)$ for all $u_1, u_2 \in P_a$ with $u_1 \ne u_2$ if all the above trace conditions hold. Thus $_c D_a F(x)=b$ has at most one solution in $P_a$ for all $a\in \Fbn$ and $b\in\Fbn$. As observed in the proof of Theorem \ref{cdu5_thm} $_c D_a F(x)=b$ has at most two solutions in $\Fbn\setminus P_a$ for all $a\in \Fbn$ and $b\in\Fbn$. Hence we have $_c\Delta_F(a,b) \le 3$ for all $a\in \Fbn$ and $b\in\Fbn$. By Theorem \ref{cdu_bound_thm}, we complete the proof.
\end{proof}

We already investigated the number of pairs $(c,\gamma)\in (\Fbn\setminus\F_2)\times (\Fbn\setminus\F_2)$ that $_c\Delta_F=3$ in Table \ref{cdu_dist_table}. In Table \ref{cdu3_table} we investigate the number of pairs $(c,\gamma)$ with $_c\Delta_F=3$ can be obtained from Corollary \ref{cdu3_coro}. Unfortunately, the number of pairs $(c,\gamma)$ that can be obtained by Corollary \ref{cdu3_coro} is only a fraction of all the pairs $(c,\gamma)$ with $_c\Delta_F=3$. However, we need to investigate all the cases that $_c D_a F(u_1)={}_c D_a F(u_2)$ where $u_1, u_2\in P_a$ with $u_1\ne u_2$ to characterize all the pairs with $_c\Delta_F=3$, which requires very routine computations.

\begin{table}
\begin{center}
\begin{tabular}{|c|c|c|c|c|c|}
\hline $n$ & 4 & 5 & 6 & 7 & 8\\
\hline \# of $(c,\gamma)$ with $_c\Delta_F=3$ & 32 & 10 & 28 & 196 & 672\\
\hline \# of $(c,\gamma)$ with $_c\Delta_F=3$ satisfying Corollary \ref{cdu3_coro} & 0 & 0 & 12 & 14 & 64 \\
\hline 
\end{tabular}
\caption{Distribution of $_c\Delta_F$ when $4\le n \le 8$.}\label{cdu3_table}
\end{center}
\end{table}

Next we characterize $c$-differential uniformity of $F$ in a special case that $\gamma\in \F_4\setminus \F_2$.

\begin{lemma}\label{af4_lemma} Let $\gamma\in \F_4\setminus \F_2$, $a\in\F_4$ and $b\in\Fbn$. Then, $_c \Delta_F(a,b)=4$ if and only if $a=1$, $b=\gamma$, $c\ne\gamma^2$ and $\tr\left(\frac{c\gamma}{c^2+\gamma}\right)=0$. Otherwise, $_c \Delta_F(a,b)\le 3$.
\end{lemma}
\begin{proof}
If $a=1$ then we get $_cD_1 F(0)=\gamma^2+c$, $_cD_1 F(1)=\ _cD_1 F(\gamma)= \gamma$ and $_cD_1 F(\gamma^2)=c\gamma$.  Since $c\ne 1$, we obtain that $\gamma^2+c$, $\gamma$ and $c\gamma$ are pairwise distinct. If $b=\gamma$ then $x=1$ and $x=\gamma$ are solutions of $_cD_1 F(x)=b$. For $x\not\in\F_4$, $_cD_1 F(x)=b$ implies $\gamma x^2+(c+\gamma^2)x+c=0$. Note that if $c\ne \gamma^2$ then $\gamma x^2+(c+\gamma^2)x+c=0$ has no solutions in $P_a=\F_4=\{0,1,\gamma,\gamma+1\}$. Hence by Lemma \ref{panot2_lemma} it has two solutions if and only if $c\ne \gamma^2$ and 
$0=\tr\left(\frac{c\gamma}{(c+\gamma^2)^2}\right)=\tr\left(\frac{c\gamma}{c^2+\gamma}\right).$ 
Hence we get $_c \Delta_F(1,\gamma)=4$ if $\tr\left(\frac{c\gamma}{c^2+\gamma}\right)=0$.\\
If $u=0$ or $u=\gamma^2$, then $_c D_1 F(x)={}_c D_1 F(u)$ has the unique solution $x=u$ in $\F_4$. By Lemma \ref{panot2_lemma}, $_c D_1 F(x)={} _c D_1 F(u)$ has at most two solutions. If $b\not \in\{\ _c D_1 F(u) : u\in\F_4\}$, then we get $_c \Delta_F(1,b)={} _c \Delta_{Inv}(1,b)\le 3$ by Theorem \ref{cdu_inv_even}. Therefore, we get $_c \Delta_F(\gamma,b)\le 3$ for all $b\in\Fbn$ with $b\ne \gamma$.

If $a=\gamma$, then we get $_cD_\gamma F(0)=c$, $_cD_\gamma F(1)= \gamma+c\gamma^2$, $_cD_\gamma F(\gamma)= 1$ and $_cD_\gamma F(\gamma^2)=\gamma^2+c\gamma$. Since $c\ne 1$, we obtain that $\gamma+c$, $c\gamma^2$, $\gamma^2$ and $1+c\gamma$ are pairwise distinct(because if two of them are same then we get $c=1$, a contradiction). Thus $_c D_\gamma F(x)=\ _c D_\gamma F(u)$ has the unique solution $x=u$ in $\F_4$ and at most two solutions in $x\in\Fbn\setminus\F_4$, by Lemma \ref{panot2_lemma}. Hence we have $_c \Delta_F(\gamma,b)\le 3$ for all $b\in\Fbn$.

If $a=\gamma^2$ then we get $_cD_{\gamma^2} F(0)=\gamma+c$, $_cD_{\gamma^2} F(1)= c\gamma^2$, $_cD_{\gamma^2} F(\gamma)= \gamma^2$ and $_cD_{\gamma^2} F(\gamma^2)=1+c\gamma$.  Since $c\ne 1$, we obtain that $\gamma+c$, $c\gamma^2$, $\gamma^2$ and $1+c\gamma$ are pairwise distinct(If two of them are same then we get $c=1$, a contradiction). Similar with the case $a=\gamma$, we have $_c \Delta_F(\gamma^2,b)\le 3$ for all $b\in\Fbn$.
\end{proof}

\begin{theorem}\label{cf4_thm} Let $n$ be even and $\gamma\in\F_4\setminus \F_2$. If $c\in\F_4 \setminus \F_2$, then
\begin{equation*}
_c\Delta_F =
\begin{cases}
3 &\text{if }n\equiv 2 \pmod{4},\\
4 &\text{if }n\equiv 4 \pmod{8},\\
5 &\text{if }n\equiv 0 \pmod{8}.
\end{cases}
\end{equation*}
\end{theorem}

\begin{proof} By Theorem \ref{cdu5_thm} we have $_c\Delta_F=5$ if and only if $n\equiv 0\pmod{8}$. Otherwise we have $3\le\ _c\Delta_F \le 4$ by Theorem \ref{cdu_bound_thm} and Theorem \ref{cdu5_thm}. By Lemma \ref{cdu_symm_lemma}, it is sufficient to consider the case $c=\gamma$.

\bigskip
\noindent (Case 1) Assume that $n\equiv 2\pmod{4}$. By Lemma \ref{af4_lemma} we have $_c\Delta_F(a,b)\le 3$ for all $a\in \F_4$ and $b\in\Fbn$, since $\tr\left(\frac{c\gamma}{c^2+\gamma}\right)=\tr(\gamma^2)=1$. Hence it is sufficient to consider the case that $a\in\Fbn\setminus \F_4$. Using the proof of Theorem \ref{cdu5_thm} we have the followings :
\begin{itemize}
\item $_c D_a F(0)\ne {} _c D_a F(1)$ since $\tr\left(\frac{\gamma(c\gamma+1)}{(\gamma+1)^2}\right)=\tr\left(\frac{\gamma}{\gamma(\gamma+1)}\right)=\tr(\gamma)=1$
\item $_c D_a F(0)\ne {} _c D_a F(\gamma)$ since $\tr(c^{-1}\gamma^{-1})=\tr(\gamma)=1$
\item $_c D_a F(0)\ne {} _c D_a F(a+1)$ since $\tr\left(\frac{\gamma(c\gamma+1)}{(\gamma+1)^2}\right)=\tr\left(\frac{\gamma(\gamma^2+1)}{(\gamma+1)^2}\right)=\tr(\gamma)=1$
\item $_c D_a F(0)\ne {} _c D_a F(a+\gamma)$ since if $_c D_a F(0)= \ _c D_a F(a+\gamma)$ then we have $a=1$, a contradiction to $a\in\Fbn\setminus \F_4$
\item $_c D_a F(1)\ne {} _c D_a F(\gamma)$ since $\tr\left(\frac{\gamma(c+\gamma+1)}{c(\gamma+1)^2}\right)=\tr\left(\frac{1}{\gamma^4}\right)=\tr(\gamma^2)=1$
\item $_c D_a F(1)\ne {} _c D_a F(a)$ since if $_c D_a F(1)= \ _c D_a F(a)$ then we have $a=\gamma^2$, a contradiction to $a\in\Fbn\setminus \F_4$
\item $_c D_a F(1)\ne {} _c D_a F(a+\gamma)$ since if $_c D_a F(1)= \ _c D_a F(a+\gamma)$ then we have $a=\gamma^2$, a contradiction to $a\in\Fbn\setminus \F_4$
\item $_c D_a F(\gamma)\ne {} _c D_a F(a)$ since $\tr\left(\frac{c\gamma}{(c+\gamma+1)^2} \right)=\tr(\gamma^2)=1$
\item $_c D_a F(\gamma)\ne {} _c D_a F(a+1)$ since $\tr\left(\frac{c\gamma^2}{(c\gamma+1)^2} \right)=\tr(\gamma)=1$
\item $_c D_a F(a)\ne {} _c D_a F(a+1)$ since if $_c D_a F(a)= \ _c D_a F(a+1)$ then we have $a^2+a+1=0$, a contradiction to $a\in\Fbn\setminus \F_4$
\item $_c D_a F(a)\ne {} _c D_a F(a+\gamma)$ since if $_c D_a F(a)= \ _c D_a F(a+\gamma)$ then we have $a^2+\gamma a+\gamma^2=0$ and then $a\in\{1,1+\gamma\}$, a contradiction to $a\in\Fbn\setminus \F_4$
\item $_c D_a F(a+1)\ne {} _c D_a F(a+\gamma)$ since if $_c D_a F(a+1)= \ _c D_a F(a+\gamma)$ then we have $a=\gamma +1$, a contradiction to $a\in\Fbn\setminus \F_4$
\end{itemize}
Hence $_c D_a F(x)=b$ has at most one solution in $P_a$. Since $_c D_a F(x)=b$ has at most two solutions in $\Fbn\setminus P_a$, we obtain $_c\Delta_F(a,b)\le 3$.\\
\noindent (Case 2) Assume that $n\equiv 4\pmod{8}$. Then by Lemma \ref{af4_lemma} we have $_c\Delta_F(1,\gamma)=4$ since $\tr\left(\frac{c\gamma}{c^2+\gamma}\right)=\tr(\gamma^2)=0$. Therefore, we have $_c\Delta_F=4$ in this case.
\end{proof}

Next we propose a sufficient condition for $_c\Delta_F=3$ in case $\gamma\in\F_4\setminus\F_2 $ using Corollary \ref{cdu3_coro}.

\begin{corollary}\label{cf4not_thm} Let $n$ be even and $\gamma\in\F_4\setminus \F_2$. If $c\in\Fbn\setminus\F_4$, then $3\le {}_c\Delta_F \le 4$. Furthermore, $_c\Delta_F=3$ if $\tr(c\gamma)=\tr(c\gamma^2)=\tr(c^{-1}\gamma)=\tr(c^{-1}\gamma^2)=\tr\left(\frac{c}{(c+\gamma)^2}\right)=\tr\left(\frac{c\gamma^2}{(c+\gamma)^2}\right)=\tr\left(\frac{c\gamma}{(c+\gamma^2)^2} \right)=\tr\left(\frac{c}{(c+\gamma^2)^2} \right)=1.$
\end{corollary}

\begin{proof} By Theorem \ref{cdu5_thm} if $\gamma\in\F_4\setminus \F_2$ then $_c\Delta_F=5$ if and only if $c\in\F_4\setminus \F_2$ and $n\equiv 0\pmod{8}$. Hence if $c\in\Fbn\setminus \F_4$ then we have $3\le {}_c\Delta_F \le 4$ using Theorem \ref{cdu_bound_thm}. If $\gamma\in\F_4\setminus \F_2$ and $c\in\Fbn\setminus \F_4$  then we have $c\not\in\{\gamma,\gamma+1,\gamma^{-1},(\gamma+1)^{-1}\}$. Using $\gamma\in\F_4\setminus \F_2$ and $\tr(1)=0$ we check all trace conditions in Corollary \ref{cdu3_coro}
\begin{equation*}
\begin{array}{ll}
\tr\left(\frac{\gamma}{c(\gamma+1)}\right)=\tr(c^{-1}\gamma^{-1})=1& \tr(c^{-1}\gamma^{-1})=1\\
\tr\left(\frac{\gamma(c\gamma+1)}{(\gamma+1)^2}\right)=\tr(c\gamma+1)=\tr(c\gamma)=1 & 
\tr\left(\frac{c\gamma}{(c\gamma+c+1)^2}\right)=\tr\left(\frac{c}{(c+\gamma)^2}\right)=1\\
\tr\left(\frac{\gamma(c+\gamma+1)}{c(\gamma+1)^2}\right)=\tr\left(\frac{c+\gamma^2}{c}\right)=\tr(c^{-1}\gamma^2)=1\ \ &
\tr\left(\frac{\gamma(c+\gamma)}{c(\gamma+1)^2}\right)=\tr\left(\frac{c+\gamma}{c}\right)=\tr(c^{-1}\gamma)=1\\
\tr\left(\frac{c\gamma^2}{(c+\gamma)^2}\right)=1&
\tr\left(\frac{c\gamma}{(c+\gamma+1)^2} \right)=\tr\left(\frac{c\gamma}{(c+\gamma^2)^2} \right)=1\\
\tr\left(\frac{c\gamma^2}{(c\gamma+1)^2} \right)=\tr\left(\frac{c}{(c+\gamma^2)^2} \right)=1&
\tr\left(\frac{c\gamma}{\gamma+1}\right)=\tr(c\gamma^2)=1\\
\tr(c\gamma^{-1})=\tr(c\gamma^2)=1&
\tr\left(\frac{(c\gamma+c+1)\gamma}{(\gamma+1)^2}\right)=\tr(c\gamma^2+1)=\tr(c\gamma^2)=1
\end{array}
\end{equation*}
and hence we have $_c\Delta_F=3$ by Corollary \ref{cdu3_coro}, which completes the proof.
\end{proof}

The trace conditions in Corollary \ref{cf4not_thm} are not necessary for $_c\Delta_F=3$. The third row of Table \ref{cf4not_table} indicates the number of $c\in \Fbn \setminus \F_4$ with $_c\Delta_F=3$ satisfying all the trace conditions in Corollary \ref{cf4not_thm} for each $4\le n\le 12$.

\begin{table}
\begin{center}
\begin{tabular}{|c|c|c|c|c|c|}
\hline $n$ & 4 & 6 & 8 & 10 & 12\\
\hline \# of $(c,\gamma)$ with $_c\Delta_F=3$ & 4 & 0 & 8 & 20 & 136\\	
\hline \# of $(c,\gamma)$ with $_c\Delta_F=3$ satisfying Corollary \ref{cf4not_thm} & 0 & 0 & 8 & 10 & 84\\
\hline 
\end{tabular}
\caption{The number of $c\in \Fbn \setminus \F_4$ when $4\le n \le 12$.}\label{cf4not_table}
\end{center}
\end{table}

\section{Concluding Remark}

In this paper, we study $c$-differential uniformity of permutations with low Carlitz rank. We show that a permutation of Carlitz rank $m$ has $c$-differential uniformity at most $m+2$. Hence we can see that a permutation of low Carlitz rank has low $c$-differential uniformity. We observe that this upper bound $m+2$ on $c$-differential uniformity of permutations with Carlitz rank $m$ is tight when $1\le m \le 3$. In particular, we investigate $c$-differential uniformity of permutations of the form $Inv\circ (0,1,\gamma)$, which have the same $c$-differential uniformity with some permutations with Carlitz rank $3$. We can see that $3\le \ _c\Delta_F \le 5$, and we characterize the case $_c\Delta_F=5$ and give a sufficient condition for $_c\Delta_F=3$. We also give a refined chracterization of $_c\Delta_F$ for the special case that $n$ is even and $\gamma\in\F_4 \setminus \F_2$. 

The proof of an upper bound on $c$-differential uniformity of permutations with Carlitz rank $m$ is based on the fact that they are affine equivalent (of degree one) to the inverse function with $m$ modified points. Since all permutations modifying a small set of points from $Inv$ has low Carlitz rank, we already show that they also have low $c$-differential uniformity. In future studies, we investigate the $c$-differential uniformity of them in detail. 

\bigskip
\noindent \textbf{Acknowledgements} : This work was supported by the National Research Foundation of Korea (NRF) grant
funded by the Korea government (MSIT) (No. 2021R1C1C2003888). Soonhak Kwon was supported by the National Research Foundation of Korea (NRF) grant funded by the Korea government (MSIT) (No. 2016R1A5A1008055, No. 2019R1F1A1058920 and 2021R1F1A1050721).

\end{document}